\documentclass[11pt]{article}
\usepackage{epsfig,amsmath,amssymb,stmaryrd,enumerate}
\usepackage{graphicx}

\usepackage{latexsym}
\usepackage{amsthm}
\usepackage{amsmath}
\usepackage{amsfonts, epsfig, amsmath, amssymb, color, amscd}
\usepackage[cmtip,arrow]{xy}
\usepackage{pb-diagram, pb-xy}
\usepackage{amssymb,epsfig}
\usepackage{color}
\usepackage[cmtip,arrow]{xy}
\usepackage{pb-diagram, pb-xy}
\usepackage{amssymb,epsfig,amsfonts}

\newfont{\sdbl}{msbm9}
\newfont{\dbl}{msbm10 at 12pt}
\theoremstyle{definition}

\newcommand{\dn}{{\mbox{\dbl N}}}

\newcommand{\dpp}{{\mbox{\dbl P}}}

\newtheorem{nt}{Remark}[section]

\theoremstyle{plain}
\newtheorem{prop}{Proposition}[section]
\newtheorem{theo}{Theorem}[section]
\newtheorem{lemma}{Lemma}[section]

\begin{document}

\title{
Commuting ordinary differential operators with polynomial
coefficients and  automorphisms of the
first Weyl algebra
}

\author{{Andrey E. Mironov\thanks{
The first author was supported by RSF (grant 14-11-00441).}
\ and
Alexander B. Zheglov
 }}

\date{}
\maketitle

\begin{abstract}
In this paper we study rank two commuting ordinary differential operators with polynomial coefficients
and the orbit space of the automorphisms group of the first Weyl algebra on such operators. We prove
that for arbitrary fixed  spectral curve of genus one the space of orbits is infinite. Moreover, we prove in this case that for for any $n\ge 1$ there is a pair of self-adjoint commuting ordinary differential
operators of rank two $L_4=(\partial_x^2+V(x))^2+W(x)$, $L_{6}$, where $W(x),V(x)$ are polynomials  of degree $n$ and $n+2$. We also prove that there are hyperelliptic spectral curves with the infinite spaces of orbits.

\end{abstract}

\section{Introduction}

The group of automorphisms  of the first Weyl algebra $A_1=\{\sum_{j=0}^nu_j(x)\partial_x^j,u_j\in{\mathbb C}[x]\}$ acts on the set of solutions of the equation
\begin{equation}\label{eq1}
 f(X,Y)=\sum_{j,i=0}^n\alpha_{ij}X^iY^j=0,\quad X,Y\in A_1,
 \alpha_{ij}\in{\mathbb C},
\end{equation}
 i.e. if $X,Y\in A_1$ satisfy (\ref{eq1}) and $\varphi\in Aut(A_1)$, then $\varphi (X),\varphi (Y)$
also satisfy (\ref{eq1}). The group $Aut(A_1)$ is generated by
the following automorphisms
$$
 \varphi_1(x)=\alpha x+\beta\partial_x,\quad \varphi_1(\partial_x)=\gamma x+\delta\partial_x,\quad \alpha,\beta,\gamma,\delta\in{\mathbb C},\quad \alpha\delta-\beta\gamma=1,
$$
$$
 \varphi_2(x)=x+P_1(\partial_x),\quad \varphi_2(\partial_x)=\partial_x,
$$
$$
 \varphi_3(x)=x,\quad \varphi_2(\partial_x)=\partial_x+P_2(x),
$$
where $P_1, P_2$ are arbitrary polynomials (see \cite{D}). So, $Aut(A_1)$ consists of tame automorphisms.
A natural and important problem is to describe the orbit space of the group action of $Aut(A_1)$ in the set of solutions of \eqref{eq1}.
If one describes the orbit space it gives a chance to compare $End(A_1)$ and $Aut(A_1)$ ($End(A_1)$ consists of endomorphisms $\varphi:A_1\rightarrow A_1$, i.e. $[\varphi(\partial_x),\varphi(x)]=1$). Let us recall the Dixmier conjecture: $End(A_1)=Aut(A_1)$, or
in other words, if differential operators $L_n, L_m$ with polynomial coefficients
satisfy the string equation
$$
 [L_n,L_m]=1,
$$
then $L_m,L_n$ can be obtained from $x,\partial_x$ with the help 
of compositions
$\varphi_j$ above (the general Dixmier conjecture for $A_n$ is stably
equivalent to the Jacobian conjecture due to \cite{KK}).
Berest has proposed the following interesting conjecture:

{\it If the Riemann surface corresponding to the equation $f=0$
 with generic $\alpha_{ij}\in{\mathbb C}$ has genus $g=1$ then the orbit space is infinite, and if $g>1$ then there are only finite number of orbits.}

One can prove that if there are finite number of orbits for some equation (\ref{eq1}) then $End(A_1)=Aut(A_1)$.

In this paper we consider the equation
\begin{equation}\label{eq2}
 Y^2=X^{2g+1}+c_{2g}X^{2g}+\dots+c_1X+c_0,\quad X,Y\in A_1,c_j\in{\mathbb C}.
\end{equation}
Using Schur's arguments \cite{Sch} one can prove that if $X,Y\in A_1$
satisfy (\ref{eq2}) then $XY=YX$. Our approach to the above
problem is based on the
 Krichever--Novikov theory
of commuting  higher rank ordinary differential operators.
Let us recall some basic notions and facts related to commuting differential
operators. If $L_n=\sum_{j=0}^nv_j(x)\partial_x^j$,
$L_m=\sum_{k=0}^mu_k(x)\partial_x^k$
commute then there is a Burchnall--Chaundy's polynomial $F(z,w)$
which vanishes the operators, $F(L_n,L_m)=0$.

The {\it spectral curve}
$\Gamma$ defined by the
equation $F=0$ is irreducible and is
completed at infinity with a unique point $q$. The spectral curve parametrizes common eigenvalues
of $L_n$ and $L_m$, i.e. if
$
 L_n\psi=z\psi,\quad L_m\psi=w\psi,
$
then $(z,w)\in\Gamma.$ The dimension of the space of common
eigenfunctions
for generic $P=(z,w)\in\Gamma$ is called the {\it rank}.  Commutative
rings of ordinary differential operators were classified by
Krichever \cite{K1}, \cite{K2}. In the case of rank one
eigenfunctions are Baker--Akhiezer functions, found by Krichever. The case of
rank $l>1$ is very complicated. In this case the
eigenfunctions can not be found explicitly. Operators of rank two
corresponding to elliptic spectral curves were found by Krichever
and Novikov \cite{KN}, operators of fourth order have the form
$$
 L_{KN}=\left(\partial_x^2+u\right)^2+2c_x(\wp(\gamma_2)-\wp(\gamma_1))\partial_x+(c_x(\wp(\gamma_2)-\wp(\gamma_1)))_x-
 \wp(\gamma_2)-\wp(\gamma_1),
$$
where
$
 \gamma_1(x)=\gamma_0+c(x), \gamma_2(x)=\gamma_0-c(x),
$
$$
 u(x)=-\frac{1}{4c_x^2}+\frac{1}{4}\frac{c_{xx}^2}{c_x^2}+2\Phi(\gamma_1,\gamma_2)c_x-\frac{c_{xxx}}{2c_x}+
 c_x^2(\Phi_c(\gamma_0+c,\gamma_0-c)-\Phi^2(\gamma_1,\gamma_2)),
$$
$$
 \Phi(\gamma_1,\gamma_2)=\zeta(\gamma_2-\gamma_1)+\zeta(\gamma_1)-\zeta(\gamma_2),
$$
$\zeta(z),\wp(z)$ are the Weierstrass functions, $c(x)$ is an arbitrary smooth function, $\gamma_0$ is a constant.
The operator $L_{KN}$ commutes with a six order differential operator
$\tilde{L}_{KN}$.

Let us formulate our main results.

\vspace{0.4cm}

\begin{theo}
\label{main1}
For arbitrary integer $m>0$ and arbitrary spectral curve $\Gamma$ given by the equation $w^2=z^3+c_2z^2+c_1z+c_0$
there are polynomials
$$
V_m=\alpha_{m+2}x^{m+2}+\ldots +\alpha_0, \mbox{\quad} W_m=\beta_mx^m +\ldots +\beta_0, \quad \alpha_{m+2}\neq 0, \beta_m\neq 0
$$
such that the operator
$$
 L_{4,m}=(\partial_x^2+V_m(x))^2+W_m(x)
$$
commutes with a six order operator $L_{6,m}$. The spectral curve of
$L_{4,m},L_{6,m}$ coincides with ñ $\Gamma$.
\end{theo}
At $m=1$ we have
$$
 L_{4,1}=(\partial_x^2+\alpha_3x^3+\alpha_2x^2+\alpha_1x+\alpha_0)^2+
 2\alpha_3x,\quad \alpha_3\ne 0.
$$
At $\alpha_3=1,\alpha_1=\alpha_2=0$ the operators $L_{4,1}, L_{6,1}$ coincide
with the Dixmier operators \cite{D}. The example of Dixmier was the first
example
of commutative subalgebra in $A_1$. It is an interesting problem
how to obtain $L_{4,m},L_{6,m}$ from $L_{KN},\tilde{L}_{KN}$?
At $m=1$ the answer is given in the Grinevich's theorem \cite{G}:

\begin{itemize}
  \item Operator $L_{KN}$ corresponding to the curve $w^2=4z^3+g_2z+g_3$ has     rational coefficients if and only if
$$
c(x)=\int_{q(x)}^{\infty}\frac{dt}{\sqrt{4t^3+g_2t+g_3}},
$$
where $q(x)$ is a rational function. If $\gamma_0=0$ and $q(x)=x$, then
$L_{KN}$ coincides with $L_{4,1}$.
\end{itemize}
Theorem 1.1 allows to prove the following theorem.

\begin{theo}
The set of orbits of the group $Aut(A_1)$
in the space of solutions of arbitrary equation
$$
 Y^2=X^{3}+c_{2}X^{2}+c_1X+c_0,\quad X,Y\in A_1,c_j\in{\mathbb C}
$$
is infinite.
\end{theo}

\vspace{0.4cm}

 Commuting operators of rank two of order $4$ and $4g+2$
 corresponding to hyperelliptic spectral curves of genus $g$ were
 studied in \cite{M1}. With the help of methods of \cite{M1} one can
 construct rank 2 operators at $g>1$. For example
$$
 L^{^{\sharp}}_4=(\partial_x^2+\alpha_3x^3+\alpha_2 x^2+\alpha_1x+\alpha_0)^2+g(g+1)\alpha_3x, \qquad \alpha_3\ne 0
$$
commutes with an operator $L^{^{\sharp}}_{4g+2}$ \cite{M1}. Mokhov
\cite{Mokh1} proved that if one apply elements of $Aut(A_1)$ to $L^{^{\sharp}}_4,L^{^{\sharp}}_{4g+2}$ then one can obtains operators
of rank $l=2k$ and $l=3k$, where $k$ is a positive integer. For example
if we apply the automorphism
$
 \varphi(x)=\partial_x, \varphi(\partial_x)=-x
$ to $L^{^{\sharp}}_4,L^{^{\sharp}}_{4g+2}$
we obtain rank 3 operators. Herewith
$$
 \varphi(L^{^{\sharp}}_4)=(\alpha_3\partial_x^3+\alpha_2 \partial_x^2+\alpha_1\partial_x+\alpha_0+x^2)^2+g(g+1)\alpha_3\partial_x.
$$
Another important example constructed in \cite{M2} is the following. The operator
$$
 L_4^{\natural}=(\partial_x^2+\alpha_1\cosh x+\alpha_0)^2+
 \alpha_1g(g+1)\cosh x,\quad
 \alpha_1\ne 0
$$
commutes with $L^{^{\natural}}_{4g+2}$.
Using $L_4^{\natural},L^{^{\natural}}_{4g+2}$ Mokhov constructed examples of
operators of arbitrary rank $l>1$ \cite{Mokh2} (we discuss this
construction in section 2).
Let $\Gamma^{\natural}$ be a spectral curve of $L_4^{\natural},L^{^{\natural}}_{4g+2}$ given by the equation
\begin{equation}
\label{hyperelliptic}
 w^2=z^{2g+1}+c_{2g}^{\natural}z^{2g}+\dots+c_{1}^{\natural}z+c_0^{\natural}.
\end{equation}
Coefficients $c_{j}^{\natural}$ can be found with the help of a recurrent
formula (see Lemma 1 bellow). Probably for all $g$ the curve
$\Gamma^{\natural}$ is not singular
for general set of parameters $\alpha_0, \alpha_1$. For small $g$
using Lemma 1 one can check this by direct calculation.

\vspace{0.4cm}

\begin{theo}
The set of orbits of the group $Aut(A_1)$ in the space of solutions
of the equation
$$
 Y^2=X^{2g+1}+c_{2g}^{\natural}X^{2g}+
 \dots+c_{1}^{\natural}X+c_0^{\natural},\quad X,Y\in A_1
$$
is infinite.
\end{theo}

It would be interesting to check the Berest conjecture at $g>1$ for
generic equation (\ref{eq1}) having a nonconstant solution in $A_1$.

\begin{nt}
The group $Aut(A_1)$ acts on the set of rings of commuting differential
operators with affine spectral curves considered in Theorems 1.2 and 1.3.
One can prove that the space of orbits is also infinite.
\end{nt}

\section{Method of deformation of Tyurin parameters}

Every ring $A$ of commuting ordinary differential operators is isomorphic
to a ring of meromorphic functions on spectral curve $\Gamma$
with a pole in some point $q\in\Gamma$
(we consider in this section the case when $\Gamma$ is nonsingular, i.e. $\Gamma$ is a
Riemann surface). For a meromorphic function $f(P)$, $P\in\Gamma$ with pole
in $q$ of order $n$ we have $L_f\psi(x,P)=f(P)\psi(x,P)$ where $L_f\in A$
is a differential operator of order $ln$, $l$ is the rank of commuting
operators, $\psi=(\psi_1,\dots,\psi_{l})$ is a vector Baker--Akhiezer function.
Function $\psi$ can be reconstructed from the following spectral data
(see \cite{K1})
$$
 \{\Gamma,q,k^{-1},\gamma_1,\dots,\gamma_{lg},\alpha_1,\dots,\alpha_{lg},
 \omega_1(x),\dots,\omega_{l-1}(x)\}.
$$
Hire $k^{-1}$ is a local parameter near $q$, $g$ is the genus of $\Gamma$,
$\gamma_j\in\Gamma$, $\alpha_j=(\alpha_{j,1},\dots,\alpha_{j,l-1})$ is a
vector,
$\omega_j(x)$ is a smooth function. The set $(\gamma,\alpha)$ is called the
{\it Tyurin parameters}. This parameters define a semi-stable holomorphic
rank $l$ vector bundle on $\Gamma$ of degree $lg$ with holomorphic
sections $\eta_1,\dots,\eta_l$. The points $\gamma_1,\dots,\gamma_{lg}$
are points of their  linear dependence of the sections
$$
 \eta_l(\gamma_j)=\sum_{i=1}^{l-1}\alpha_{j,i}\eta_i(\gamma_j).
$$
The vector-function $\psi$ is defined by the following properties.

1. In the neighbourhood of $q$ it has the form
$$
 \psi(x,P)=\left(\sum_{s=0}^{\infty}\xi_s(x)k^{-s}\right)\Phi(x,k),
$$
where $\xi_0=(1,0,\dots,0), \xi_i(x)=(\xi_i^1(x),\dots,\xi_i^l(x))$,
the matrix  $\Phi$ satisfies the equation
$$
 \frac{d\Phi}{dx}=A\Phi,\qquad
A=\left(
\begin{array}{cccccc}
 0 & 1 & 0 & \dots  & 0 & 0\\
 0 & 0 & 1 & \dots  & 0 & 0\\
 \dots & \dots & \dots & \dots &  \dots & \dots\\
 0 & 0 & 0 & \dots  & 0 & 1 \\
 k+\omega_1 & \omega_2 & \omega_3 & \dots & \omega_{l-1} & 0
 \end{array}\right).
$$
2.  The components of $\psi$ are meromorphic functions on
$\Gamma\backslash\{q\}$ with the simple poles $\gamma_1,\dots,\gamma_{lg}$,
and
$$
 {\rm Res}_{\gamma_i}\psi_j=\alpha_{i,j}{\rm Res}_{\gamma_i}\psi_{l},\quad
 1\leq i\leq lg,\ 1\leq j\leq l-1.
$$
The main difficulty to construct operators of rank $l>1$ is the fact that the Baker--Akhiezer function is not found explicitly.  In the recent paper \cite{MS}
were shown that the class of Baker--Akhiezer functions contains
some known special functions.

Let us recall the method of deformation of Tyurin parameters \cite{KN}.
The main idea of this method is to study the linear differential operator
which vanishes the common eigenfunctions.
The common eigenfunctions of commuting differential operators of rank $l$ satisfy the linear differential equation of order $l$
$$
 \psi^{(l)}(x,P)=\chi_0(x,P)\psi(x,P)+\dots+\chi_{l-1}(x,P)\psi^{(l-1)}(x,P).
$$
The coefficients $\chi_i$ are rational functions on $\Gamma$ with the simple poles
$P_1(x),$ $\dots,$ $P_{lg}(x)\in\Gamma$, and with the following expansions in the
neighbourhood of $q$
$$
\chi_0(x,P)=k+g_0(x)+O(k^{-1}),\qquad
 \chi_j(x,P)=g_j(x)+O(k^{-1}), \ \
 0<j<l-1,
$$
$$
\chi_{l-1}(x,P)=O(k^{-1}).
$$
Let $k-\gamma_i(x)$ be a local parameter near $P_i(x)$. Then
$$
\chi_j=\frac{c_{i,j}(x)}{k-\gamma_i(x)}+d_{i,j}(x)+O(k-\gamma_i(x)).
$$
Functions $c_{ij}(x) ,d_{ij}(x)$ satisfy the following equations (see \cite{K1}).

\begin{equation}\label{w1}
 c_{i,l-1}(x)=-\gamma'_i(x),
\end{equation}
\begin{equation}\label{w2}
 d_{i,0}(x)=\alpha_{i,0}(x)\alpha_{i,l-2}(x)+\alpha_{i,0}(x)d_{i,l-1}(x)
 -\alpha'_{i,0}(x),
\end{equation}
\begin{equation}\label{w3}
 d_{i,j}(x)=\alpha_{i,j}(x)\alpha_{i,l-2}(x)-\alpha_{i,j-1}(x)+\alpha_{i,j}(x)d_{i,l-1}(x)
 -\alpha'_{i,j}(x), j\geq 1,
\end{equation}
where
$
 \alpha_{i,j}(x)=\frac{c_{i,j}(x)}{c_{i,l-1}(x)}, \ \ 0\leq j\leq l-1, \ 1\leq i\leq lg.
$
To find $\chi_i$ one should solve the equations (\ref{w1})--(\ref{w3}). Using $\chi_i$ one can find coefficients of the operators.
At $g=1$, $l=2$ Krichever and Novikov \cite{KN} solved these equations and found the operators $L_{KN}$. Operators of Krichever--Novikov and it applications were
studied in \cite{Gr}--\cite{Deh}
Operators of rank 3 corresponding to elliptic spectral  curves were found by Mokhov \cite{Mokh}.
In \cite{M3}--\cite{Z} some examples of operators of rank 2,3
 corresponding to spectral curves of genus 2--4 were constructed.

In \cite{M1} commuting operators of rank two of order $4$ and $4g+2$ corresponding to hyperelliptic spectral curves were studied
$$
 L_4\psi=z\psi,\quad L_{4g+2}\psi=w\psi,\quad w^2=F_g(z)=z^{2g+1}+c_{2g}z^{2g}+\dots+c_0.
$$
Common eigenfunctions of $L_4$ and $L_{4g+2}$ satisfy the second order differential equation
$$
 \psi''-\chi_1(x,P)\psi'-\chi_0(x,P)\psi=0,\ P=(z,w)\in\Gamma,
$$
where $\chi_0(x,P), \chi_1(x,P)$ are rational functions on $\Gamma$ satisfying  equations (\ref{w1})--(\ref{w3}).

 ${\ }$

\noindent{\bf Theorem 2} (\cite{M1}) {\it  The operator $L_4$ is formally
self-adjoint if and only if
$$
 \chi_1(x,P)=\chi_1(x,\sigma(P)),
$$
where $\sigma$ is the hyperelliptic involution on $\Gamma$.
}

\vspace{0.7cm}

\noindent{\bf Theorem 3} (\cite{M1}) {\it If $L_4$ is formally self-adjoint,
i.e. $L_4=(\partial_x^2+V(x))^2+W(x),$
then
$$
 \chi_0=-\frac{1}{2}\frac{Q_{xx}}{Q}+\frac{w}{Q}-V, \qquad
 \chi_1=\frac{Q_x}{Q},
$$
where
$
 Q=z^g+a_{g-1}(x)z^{g-1}+\dots+a_0(x),\ a_0(x),\dots,a_{g-1}(x)
$ are some functions.
The function $Q$ satisfies the equation
\begin{equation}\label{e1}
 4F_g(z)=4(z-W)Q^2-4V(Q_x)^2+(Q_{xx})^2-2Q_xQ_{xxx}+2Q(2V_xQ_x+4VQ_{xx}+
 \partial_x^4Q).
\end{equation}
 }

 \vspace{0.4cm}

\noindent From Theorem 3 it follows

 \vspace{0.4cm}

 \noindent{\bf Corollary 1}  {\it The function $Q$
satisfies the linear equation
\begin{equation}\label{r1}
  \partial_x^5Q+4VQ_{xxx}+6V_xQ_{xx}+2(2z-2W+V_{xx})Q_x-2W_xQ=0.
\end{equation}
}

\vspace{0.4cm}

\noindent{\bf Corollary 2}  {\it If $g=1$ then
\begin{equation}
\label{**}
 V=\frac{-16F_1(\frac{1}{2}(-c_2-W))+W_{xx}^2-2W_xW_{xxx}}{4W_x^2},
\end{equation}
where $F_1$ defines the spectral curve $w^2=F_1(x)=z^3+c_2z^2+c_1z+c_0.$
}
\vspace{0.4cm}

With the help of Theorem 3 many examples of rank 2 operators were recently
constructed (see \cite{D1}--\cite{O}).

Let us consider commuting operators
$
 L_4^{\natural},
  L_{4g+2}^{\natural}$ \cite{M2}. These operators do not commute with
operators of odd orders \cite{DS}, hence these operators are operators
of true rank 2.
The polynomial $Q$ for $L_4^{\natural}, L_{4g+2}^{\natural}$ has the form
(see \cite{M2})
$$
 Q(x,z)=A_g(z)\cosh^g x+\dots+A_1(z)\cosh x+A_0(z),
$$
where
$$
 A_s=\frac{1}{8(2s+1)\alpha_1(g(g+1)-s(s+1))}
 \left(4A_{s+5}\frac{(s+5)!}{s!}
 -8A_{s+3}\frac{(s+3)!}{s!}(2\alpha_0+s^2+4s+5)-\right.
$$
\begin{equation}\label{a1}
 \left.-8A_{s+2}\frac{(s+2)!}{s!}(2s+3)\alpha_1
 +4A_{s+1}(s+1)((s+1)^2(4\alpha_0+(s+1)^2+4z)\right),\ 0\leq s<g,
\end{equation}
we assume that $A_s=0$ at $s<0$ and $s>g$, $A_g$ is a constant.
 
\vspace{0.4cm}

\noindent{\bf Lemma 1}  (\cite{M1}) {\it The spectral curve $\Gamma^{\natural}$ 
of $L_4^{\natural}, L_{4g+2}^{\natural}$ is given by the equation
$$
 w^2=F_g(z)=\frac{1}{4}\left(4A_0^2z-4A_0A_1\alpha_1-16A_2(\alpha_0+1)+
48A_4)+4\alpha_0A_1^2+4A_2^2-2A_1(6A_3-A_1)\right),
$$
where $A_j(z)$ are defined in (\ref{a1}).
}

\vspace{0.4cm}
\noindent {\bf Examples:}

\noindent {\bf 1)} $g=1$
$$
 F_1(z)=z^3+(\frac{1}{2}-2\alpha_0)z^2+\frac{1}{16}(1-8\alpha_0+16\alpha_0^2-16\alpha_1^2)z+\frac{\alpha_1^2}{4}.
$$
\noindent {\bf 2)} $g=2$, let for simplicity of formulas $\alpha_0=0$
$$
 F_2(z)=z^5+\frac{17}{2}z^4+\frac{1}{16}(321-336\alpha_1^2)z^3+
 \frac{1}{4}(34-531\alpha_1^2)z^2+
 (1-189\alpha_1^2+108\alpha_1^4)z+24\alpha_1^2+513\alpha_1^4.
$$
The spectral curves defined by the above equations are not singular
for the general parameters.

Mokhov \cite{Mokh2} found a remarkable change of variable
$$
 x=\ln (y+\sqrt{y^2-1})^r,\quad r=\pm 1,\pm 2,\dots,
$$
which reduces the operators $L_4^{\natural}, L_{4g+2}^{\natural}$
to the operators with polynomial coefficients. In particular,
$L_4^{\natural}$ in new variable $y$ gets the form
$$
 L_4^{\natural}=((1-y^2)\partial_y^2-3y\partial_y+aT_r(y)+b)^2-ar^2g(g+1)T_r(y),
 \quad a\ne 0,
$$
$b$ is arbitrary constant, $T_r(y)$ is the Chebyshev polynomial of
degree $|r|$. Recall that
$$
 T_0(y)=1,\quad T_1(y)=y,\quad T_r(y)=2yT_{r-1}(y)-T_{r-2}(y),\quad
 T_{-r}(y)=T_r(y).
$$
Chebyshev polynomials are commuting polynomials, i.e.
$$T_n(T_m(y))=T_m(T_n(y))=T_{n+m}(y).$$
If one applies the automorphism
$$\varphi(y)=-\partial_y,\qquad \varphi(\partial_y)=y,
\quad \varphi\in Aut(A_1)$$
to the  operators $L_4^{\natural},L_{4g+2}^{\natural}$ written in $y$
variable, then
one gets operators of orders $2r, (2g+1)r$ of rank $r$ \cite{Mokh2} and
$$
 \varphi(L_{4}^{\natural})=(aT_r(\partial_y)-y^2\partial_y^2-3y\partial_y+y^2+b)^2-
arg(g+1)T_r(\partial_y).
$$

\section{Proof of Theorems 1.1--1.3}

\subsection{Proof of theorem \ref{main1}}

Let us rewrite (\ref{**}) in the form
\begin{equation}\label{eq11}
4W_x^2V=-16 F_1(\frac{1}{2}(c_2-W))+W_{xx}^2-2W_xW_{xxx},
\end{equation}
 Note that from (\ref{eq11}) it follows
\begin{equation}
\label{*}
 -4F'_1(\frac{1}{2}(-c_2-W))+2V_xW_x+4VW_{xx}+W_{xxxx}=0.
\end{equation}
Further we assume that $V,W$ are polynomials
\begin{equation}
\label{eq13}
 V=\alpha_nx^n+\ldots +\alpha_0, \mbox{\quad} W=\beta_mx^m +\ldots +\beta_0,
 \quad \alpha_n\neq 0,\ \beta_m\neq 0.
\end{equation}
Equation (\ref{eq11}) is equivalent to the system of equations: equation
(\ref{*}) and the equation on free terms of (\ref{eq11})
which is
\begin{equation}
\label{cubic}
\alpha_0\beta_1^2=-4F_1(\frac{1}{2}(c_2-\beta_0))+\beta_2^2-3\beta_1\beta_3.
\end{equation}
Let us prove the following important proposition.

\begin{prop}
\label{main}
For any $m>0$ there exists a solution
of the equation \eqref{eq11} of the form (\ref{eq13}), where $n=m+2$.
\end{prop}

\begin{proof}
Equation \eqref{*} is equivalent to a system of $2m+1$ equations in $2m+4$
variables $\alpha_i,\beta_j$. Note that all equations have degree 2 and the
set of their solutions consists of points in ${\mathbb C}^{2m+4}$
(with coordinates $\alpha_i,\beta_j$) which lie in the intersection of
$2m+1$ quadrics defined by these equations. By \cite[Ch.1,Th.7.2]{Ha}
the intersection $X$ of these quadrics in $\dpp^{2m+4}$ (with homogeneous
coordinates $\alpha_i,\beta_j,u$) is non-empty and each its irreducible
component has dimension greater or equal to 3. By the same reason the
intersection of $X$ with the hyperplane $Z=\{u=0\}$ at infinity is non-empty
and each its irreducible component has dimension greater or equal to 2.

To prove the proposition it is sufficient to prove that for
any fixed $m>0$  there is a two-dimensional irreducible component of
$X\cap Z$. From this fact we can conclude that affine part of the intersection
of quadrics is non-empty.

The homogeneous parts  of our equations in $\dpp^{2m+4}$ not depending on $u$ can be easily written: these are exactly the coefficients at $x^i$ of the sum
\begin{equation}
\label{triv1}
4VW_{xx}+2V_xW_x-3W^2.
\end{equation}
Let us introduce the following notations:
$$
V_xW_x=\sum_{i=0}^{m+n-2}b_ix^i, \mbox{\quad} VW_{xx}=\sum_{i=0}^{m+n-2}c_ix^i,
\mbox{\quad} W^2=\sum_{i=0}^{m+n-2}d_ix^i.
$$
Then the intersection $X\cap Z$ is given by the equations
\begin{equation}
\label{triv2}
4c_i+2b_i-3d_i=0, \mbox{\quad} i=0,\ldots, 2m.
\end{equation}
Note that the coefficients $b_i,c_i,d_i$ can be written in the following form:
\begin{equation}
\label{triv3}
d_i=\sum_{k=0}^i\beta_{i-k}\beta_k, \mbox{\quad} b_i=\sum_{k=0}^iB_{k,i}\alpha_{i-k+1}\beta_{k+1}, \mbox{\quad} c_i=\sum_{k=0}^i C_{k,i}\alpha_{i-k}\beta_{k+2},
\end{equation}
where $B_{k,i}=(k+1)(i-k+1),C_{k,i}=(k+1)(k+2)$ are positive integers, and we set $\alpha_j\equiv 0$ if $j>n$, $\beta_j\equiv 0$ if $j>m$.

The next observation is: equations \eqref{triv2}, \eqref{triv3} always have
a solution of the form
$$
 P=(\alpha_n\ne 0:\beta_m\neq 0:0: \ldots :0)
$$ for
any $m>0$. Indeed, if
$\alpha_0=\ldots =\alpha_{n-1}=\beta_0=\ldots =\beta_{m-1}=0$, then only
$2m$-th equation from \eqref{triv2} remains to be non-trivial, and this
equation becomes a quadratic homogeneous equation linear in $\alpha_n$ and
quadratic in $\beta_m$:
$$
(2B_{m-1,2m}+4C_{m-2,2m})\alpha_n\beta_m-3\beta_m^2=0.
$$
Thus, we can set $\beta_m=1$ where from $\alpha_n=3/(2B_{m-1,2m}+4C_{m-2,2m})$.

Let us prove that for
any fixed $m>0$  any irreducible component of $X\cap Z$ containing
$P$ has dimension 2.

If $m=1$ then there are only 3 equations in \eqref{triv2}
$$
 4C_{0,0}\alpha_0\beta_2+2B_{0,0}\alpha_1\beta_2-3\beta_0^2=0,
$$
$$
 4(C_{0,1}\alpha_1\beta_2+C_{1,2}\alpha_0\beta_3)+
 2(B_{0,1}\alpha_2\beta_1+B_{1,1}\alpha_1\beta_2)-6\beta_0\beta_1=0,
$$
$$
 4(C_{0,2}\alpha_2\beta_2+C_{1,2}\alpha_1\beta_3+C_{2,2}\alpha_0\beta_4)+
 2(B_{0,2}\alpha_3\beta_1+B_{1,2}\alpha_2\beta_2+B_{2,2}\alpha_1\beta_3)-
 3(2\beta_0\beta_2+\beta_1^2)=0,
$$
and their Jacobi matrix at $P$ has the following form:
$$
\left (
\begin{array}{cccccc}
* & 2B_{0,0}\beta_1 & 0 &0 &*&*\\
*&*& 2B_{0,1}\beta_1 &0& *&*\\
*&*&*& 2B_{0,2}\beta_1 &*&*\\
\end{array}
\right ),
$$
where the first columns denote derivations with respect to
$\alpha_0,\dots,\alpha_3$, and the last two columns denote derivations with
respect to $\beta_0,\beta_1$. The rank of the matrix is 3, so, these
equations define a smooth variety in the neighbourhood of the point $P$ of
dimension two.

For generic $m$ the point $P$ might not be regular. Nevertheless, any irreducible component containing $P$ has a dense subset of smooth points. At any such point $Q$ the Jacobi matrix $J$ can be written in the following form. It can be divided in two blocks: one consists of $m+1$ columns (derivations of equations with respect to  $\beta_0,\ldots ,\beta_m$), and another one consists of $n+1$ columns (derivations of equations with respect to $\alpha_n,\ldots , \alpha_0$).
We shall  describe only essential columns for us.

The columns of the first block are (to save the space we shall
write them as rows):
$$
\mbox{2-nd column:\quad }(j_{0,0}\alpha_1, (j_{0,1}\alpha_2-6\beta_0), (j_{0,2}\alpha_3-6\beta_1), \ldots , (j_{0,n-1}\alpha_{n}-6\beta_{m}), 0,\ldots ,0)
$$
$$
\mbox{3-d column:\quad } (j_{1,0}\alpha_0, j_{1,1}\alpha_1, (j_{1,2}\alpha_2-6\beta_0), \ldots, (j_{1,n}\alpha_{n}-6\beta_{m}), 0,\ldots ,0)
$$
$$
\mbox{4-th column:\quad } (0, j_{2,1}\alpha_0, j_{2,2}\alpha_1, (j_{2,3}\alpha_2-6\beta_0), \ldots , (j_{2,n+1}\alpha_{n}-6\beta_{m}), 0,\ldots ,0)
$$
$$
\mbox{5-th column:\quad } (0,0, j_{3,2}\alpha_0, j_{3,3}\alpha_1, (j_{3,4}\alpha_2-6\beta_0), \ldots , (j_{3,n+2}\alpha_{n}-6\beta_{m}), 0,\ldots ,0)
$$
$$
\ldots\qquad \qquad\ldots\qquad\qquad\ldots
$$
$$
\mbox{(m+1)-th column:\quad } (0,\ldots , 0, j_{m-1,m-2}\alpha_0, j_{m-1,m-1}\alpha_1, (j_{m-1,m}\alpha_2-6\beta_0), \ldots , (j_{m-1,2m}\alpha_{n}-6\beta_{m}));
$$
the columns of the second block are:
$$
\mbox{1-st column:\quad } (0,\ldots , 0, j_{0,m+1}\beta_1, j_{1,m+2}\beta_2, \ldots , j_{m-1,2m}\beta_m)
$$
$$
\mbox{2-nd column:\quad } (0,\ldots , 0, j_{0,m}\beta_1, j_{1,m+1}\beta_2, \ldots , j_{m-1,2m-1}\beta_m,0)
$$
$$
\mbox{3-d column:\quad } (0,\ldots , 0, j_{0,m-1}\beta_1, j_{1,m}\beta_2, \ldots , j_{m-1,2m-2}\beta_m,0,0)
$$
$$
\ldots\qquad \qquad\ldots\qquad\qquad\ldots
$$
$$
\mbox{n-th column:\quad } (j_{0,0}\beta_1, j_{1,1}\beta_2, \ldots , j_{m-1,m-1}\beta_m,0,\ldots ,0)
$$
$$
\mbox{(n+1)-th column:\quad } (j_{1,0}\beta_2, \ldots , j_{m-1,m-2}\beta_m,0,\ldots ,0),
$$
where the numbers $j_{k,l}$ are defined as
$$j_{k,l}=2B_{k,l}+4C_{k-1,l},$$
where we assume $B_{k,l}=0$ if $k>l$ and $C_{k,l}=0$ if $k<0$.

Without loss of generality we can assume that the point $Q$ belongs to a sufficiently small neighbourhood of the point $P$ (in the complex topology), such that, for fixed numbers $j_{k,l}$, the modules of all terms of the matrix $J$, except the terms containing $\beta_m=1$ and $\alpha_n$, are comparable with some $0<\epsilon \ll 1$ (i.e. they are $<\epsilon$ but
$>\epsilon^2$). We call such terms comparable with $\epsilon$.

We have the following possibilities now. If there is a smooth point $Q$ such that its  coordinate $\alpha_0\neq 0$ or $\alpha_1\neq 0$, then the rank of the matrix $J$ is $2m+1$, i.e. the dimension of the component is two. Indeed, we can first apply the Gauss elimination algorithm to kill all terms of the right part of the matrix lying over terms containing $\beta_m$. We can choose $\epsilon$ small enough such that the terms of the left part of $J$ will change, but the top non-zero elements of the first $m-2$ rows will remain non-zero and comparable with $\epsilon$, and all elements over them will be comparable with $\epsilon^2$. Applying again the Gauss elimination algorithm we can kill all elements in the columns except these top non-zero elements, thus obtaining $2m+1$ linearly independent rows in the matrix $J$.

Note that the case $\alpha_0=0$, $\alpha_1\neq 0$ (i.e. $\alpha_0=0$ for all smooth points) is in fact impossible: in this case the whole component belongs to the hyperplane $\alpha_0=0$. But then the dimension of the component must be 1, a contradiction.

Now we claim that there {\it exists} a smooth point such that $\alpha_0\neq 0$ or $\alpha_1\neq 0$. Indeed,
if there are no such smooth points, then the whole component belongs to the intersection of hyperplanes $\alpha_0=\alpha_1=0$ (cf. \cite[Ch.1,ex.1.6]{Ha}). Note that in this case from 0-th equation in \eqref{triv2} it follows $\beta_0=0$, and from the $1$-st equation it follows $\alpha_2\beta_1=0$.

Let's show first that $\alpha_2=\beta_1=0$.
If there is a smooth point in the component with $\alpha_2\neq 0$, then the Jacobi matrix of our system restricted to the $2m$-dimensional intersection of hyperplanes $\alpha_0=\alpha_1=\beta_0=0$ reduces to the following matrix.

The columns of the first block are (to save the space we will again write them as rows):
$$
\mbox{1-st column:\quad }(j_{0,1}\alpha_2, (j_{0,2}\alpha_3-6\beta_1),  \ldots , (j_{0,n-1}\alpha_{n}-6\beta_{m}), 0,\ldots ,0)
$$
$$
\mbox{2-nd column:\quad } (0, j_{1,2}\alpha_2, (j_{1,3}\alpha_3-6\beta_1), \ldots, (j_{1,n}\alpha_{n}-6\beta_{m}), 0,\ldots ,0)
$$
$$
\mbox{3-d column:\quad } (0, 0, j_{2,3}\alpha_2,  \ldots , (j_{2,n+1}\alpha_{n}-6\beta_{m}), 0,\ldots ,0)
$$
$$
\ldots\qquad \qquad\ldots\qquad\qquad\ldots\ldots
$$
$$
\mbox{m-th column:\quad } (0,\ldots , 0, j_{m-1,m}\alpha_2, \ldots , (j_{m-1,2m}\alpha_{n}-6\beta_{m}));
$$
the columns of the second block are:
$$
\mbox{1-st column:\quad } (0,\ldots , 0, j_{0,m+1}\beta_1, j_{1,m+2}\beta_2, \ldots , j_{m-1,2m}\beta_m)
$$
$$
\mbox{2-nd column:\quad } (0,\ldots , 0, j_{0,m}\beta_1, j_{1,m+1}\beta_2, \ldots , j_{m-1,2m-1}\beta_m,0)
$$
$$
\mbox{3-d column:\quad } (0,\ldots , 0, j_{0,m-1}\beta_1, j_{1,m}\beta_2, \ldots , j_{m-1,2m-2}\beta_m,0,0)
$$
$$
\ldots\qquad \qquad\ldots\qquad\qquad\ldots
$$
$$
\mbox{(n-1)-th column:\quad } (j_{0,1}\beta_1, \ldots , j_{m-1,m}\beta_m,0,\ldots ,0),
$$
where the  $m$ columns of the first block denote derivations with respect to $\beta_1,\ldots ,\beta_m$, and the $n-1$ columns of the second block denote derivations with respect to $\alpha_n,\ldots , \alpha_2$.

Since $\beta_1$ must be equal to zero, we can apply the same arguments as above and obtain that the rank of this matrix is $2m$. But this is impossible, because the dimension of the component is not less than two.

If the whole component belongs to the intersection $Y=\{\alpha_0=\alpha_1=\alpha_2=\beta_0=0\}$, but there are smooth points with $\beta_1\neq 0$, then the $2$-th equation in \eqref{triv2} reduces to $2B_{0,2}\alpha_3\beta_1 -3\beta_1^2=0$, where from we see that $\alpha_3=3\beta_1/(2B_{0,2})\neq 0$. In this case analogously to the previous case the matrix $J$ reduces to the following matrix.

The columns of the first block are:
$$
\mbox{1-st column:\quad }((j_{0,2}\alpha_3-6\beta_1),  \ldots , (j_{0,n-1}\alpha_{n}-6\beta_{m}), 0,\ldots ,0)
$$
$$
\mbox{2-nd column:\quad } (0, (j_{1,3}\alpha_3-6\beta_1), \ldots, (j_{1,n}\alpha_{n}-6\beta_{m}), 0,\ldots ,0)
$$
$$
\ldots\qquad \qquad\ldots\qquad\qquad\ldots\ldots
$$
$$
\mbox{m-th column:\quad } (0,\ldots , 0, (j_{m-1,m+1}\alpha_3-6\beta_1), \ldots , (j_{m-1,2m}\alpha_{n}-6\beta_{m}));
$$
the columns of the second block are:
$$
\mbox{1-st column:\quad } (0,\ldots , 0, j_{0,m+1}\beta_1, j_{1,m+2}\beta_2, \ldots , j_{m-1,2m}\beta_m)
$$
$$
\mbox{2-nd column:\quad } (0,\ldots , 0, j_{0,m}\beta_1, j_{1,m+1}\beta_2, \ldots , j_{m-1,2m-1}\beta_m,0)
$$
$$
\mbox{3-d column:\quad } (0,\ldots , 0, j_{0,m-1}\beta_1, j_{1,m}\beta_2, \ldots , j_{m-1,2m-2}\beta_m,0,0)
$$
$$
\ldots\qquad \qquad\ldots\qquad\qquad\ldots
$$
$$
\mbox{(n-2)-th column:\quad } (j_{0,2}\beta_1, \ldots , j_{m-1,m+1}\beta_m,0,\ldots ,0),
$$
where the  $m$ columns of the first block denote derivations with respect to $\beta_1,\ldots ,\beta_m$, and the $n-1$ columns of the second block denote derivations with respect to $\alpha_n,\ldots , \alpha_3$.

Now the situation differs from the first main case.
If we apply the Gauss elimination algorithm to kill all terms of the right part of the matrix lying over terms containing $\beta_m$, we can destroy the top non-zero terms. So, we must control the changes of these terms modulo $\epsilon^2$. Fortunately, it is not difficult: the term $j_{m-1-k,n-1-k}\alpha_3-6\beta_1$, where $0\le k\le m-1$, will be changed to the term
$$
j_{m-1-k,n-1-k}\alpha_3-6\beta_1 -\frac{j_{0,m+1-k}}{j_{m-1,2m-k}}(j_{m-1-k,2m-k}\alpha_n-6)\beta_1=
$$
$$
\frac{(-1+m)^2 m \left(2+5 m+2 m^2\right)+2 k^2 \left(-1+m^3\right)+k \left(4+m-m^3-4 m^4\right)}{m^3}.
$$
As it can be easily checked, the numerator can be equal to zero only for $k>m-1$. Thus, the rank of $J$ is equal to $2m-1=\dim Y$, a contradiction.

Now we can use the induction: suppose we have proved that
the whole component belongs to the intersection $Y=\{\alpha_0=\ldots =\alpha_{l-1}=0=\beta_0=\ldots =\beta_{l-2}\}$. Then the $2(l-1)-1$-th equation in  \eqref{triv2} implies $\alpha_l\beta_{l-1}=0$. If there is a smooth point with $\alpha_l\neq 0$, then we can apply the arguments from the first main case to show that the matrix $J$ has the maximal rank equal to the dimension of $Y$, a contradiction. If there is a smooth point with $\beta_{l-1}\neq 0$, then from $2(l-1)$-th equation we get
$$
\alpha_{l+1}=\frac{3}{j_{l-2,2(l-1)}}\beta_{l-1},
$$
and, analogously to the case $\alpha_2=0$, $\beta_1\neq 0$, we can control  the changes of the top non-zero terms $(j_{m-1-k,m-1-k+l}\alpha_{l+1}-6\beta_{l-1})$, $l\le m$. They will be changed to the terms
$$
j_{m-1-k,m-1-k+l}\alpha_{l+1}-6\beta_{l-1}-\frac{j_{l-2,m-1-k+l}}{j_{m-1,2m-k}}(j_{m-1-k,2m-k}\alpha_n-6)\beta_{l-1}=
$$
$$
\frac{1}{(-1+l)^2 m^3}(-1+l-m) (-4 k-2 k^2+12 k l+4 k^2 l-12 k l^2-2 k^2 l^2+4 k l^3-2 m+5 k m+2 k^2 m+6 l m
$$
$$
-10 k l m-2 k^2 l m-6 l^2 m+5 k l^2 m+2 l^3 m+
$$
$$
3 m^2-5 k m^2-2 k^2 m^2-6 l m^2+5 k l m^2+3 l^2 m^2+3 m^3+4 k m^3-3 l m^3-2 m^4).
$$
 The last expression is equal to zero only for $k=-2+2 l+m >m$ or
 $$
k= m\frac{-1+2 l -l^2 +m-l m+2 m^2}{2 \left(1-2 l+l^2-m+l m+m^2\right)}.
 $$
But the last expression can not be integer. Indeed, the great common divisor of $m$ and $\left(1-2 l+l^2-m+l m+m^2\right)$ must divide also the numerator, i.e. the doubled fraction must be integer. On the other hand, it is clear that the fraction is positive and less than one. It also easy to check that it can not be equal to $1/2$.

At the end we obtain that the whole component belongs to the intersection $Y=\{\alpha_0=\ldots =\alpha_{m}=0=\beta_0=\ldots =\beta_{m-1}\}$ with $\dim Y=2$. Then from $(2m-1)$-th equation we obtain $\alpha_{m+1}\equiv 0$, i.e. the component lies in $Y\cap\{\alpha_{m+1}= 0\}$, whose dimension is one, a contradiction.
\end{proof}

Let us prove Theorem 1.1
 The intersection $X'$ (in $\dpp^{2m+4}$) of $X$ from proposition \ref{main}
 and the cubic defined by (\ref{cubic}) is again non-empty, and each
 its irreducible
 component has dimension greater or equal to 2; the intersection
 $X'\cap Z$ with $Z$ is non-empty and each its irreducible component has
 dimension greater or equal to 1. The homogeneous part of (\ref{cubic})
 not depending on $u$ is
\begin{equation}
\label{***}
\alpha_0\beta_1^2+\beta_0^3/2.
\end{equation}
It also has a solution of the form $P$ from proposition \ref{main}.

To prove Theorem 1.1  it is sufficient to prove that for any fixed $m>0$
any irreducible component of $X'\cap Z$ containing $P$ has dimension 1.

Note that if $\alpha_0\neq 0$, then either $\beta_1$ or $\beta_0$ is not
equal to 0. Indeed, if $\beta_0=\beta_1=0$, then from 0-th equation it
follows that $\beta_2=0$, from 1-st equation it follows that $\beta_3=0$
and, by iteration, $\beta_m=0$, a contradiction.

Let $Q$ be a smooth point on some irreducible component of $X'\cap Z$ as in the proof of proposition \ref{main}. Consider the new Jacobi matrix with the first row consisting of partial derivatives of the equation (\ref{***}):
$$
(3\beta_0^2, 2\alpha_0\beta_1,0, \ldots, 0, \beta_1^2).
$$
If $\alpha_0\neq 0$, then it's easy to see that this row and all other rows of the old matrix $J$ are linearly independent, i.e. the dimension of the component is one.

If $\alpha_0=0$, we can literally repeat the arguments from the proof of proposition \ref{main}. Indeed, as we have already seen, in this case even an irreducible component of $X\cap Z$ would be of dimension less or equal to 1. Theorem 1.1 is proved.

\subsection{Proof of Theorems 1.2, 1.3}

According to Theorem 1.1 an arbitrary equation
$$
 Y^2=X^3+c_2X^2+c_1X+c_0,\quad X,Y\in A_1
$$
has infinitely many solutions of the form
$L_{4,m}=(\partial_x^2+V_m(x))^2+W_m(x),\ L_{6,m},$ and the equation
$$
 Y^2=X^{2g+1}+c_{2g}^{\natural}X^{2g}+
 \dots+c_{1}^{\natural}X+c_0^{\natural},\quad X,Y\in A_1
$$
also has infinitely many solutions of the form
$\varphi^{\natural}(L^{\natural}_4),\ \varphi^{\natural}(L^{\natural}_{4g+2}),$
where
$$
 L(r)=\varphi^{\natural}(L^{\natural}_4)=
 ((1-y^2)\partial_y^2-3y\partial_y+aT_r(y)+b)^2-ar^2g(g+1)T_r(y),
$$
$T_r(y)$ is the Chebyshev polynomial of degree $|r|$ (see section 2).
To prove Theorem 1.2 and Theorem 1.3 it is enough to prove that at
$r>10$ and $r\ne r_1$
$$
 \varphi(L_{4,r})\ne L_{4,r_1},\qquad
 \varphi(L(r))\ne L(r_1),
$$
for arbitrary $\varphi\in Aut(A_1)$. This facts follow from the
following lemma.

\begin{lemma}
\label{automs}
Consider a family of operators of order four with polynomial coefficients
$$
 L(r)=(a(x)\partial_x^2+b(x)\partial_x+c_r(x))^2+d_r(x), \mbox{\quad }
 r\in \dn,
$$
where $a(x), b(x)$ are polynomials of fixed degree such that
$$
 {\rm deg} a(x)>{\rm deg} b(x),\quad {\rm deg} c_r(x)=r,
 \quad r\geq {\rm deg}  d_r(x).
$$
If $r>{\rm deg} a(x)+8$, then
$$
 \varphi(L(r))\ne L(r_1)
$$
at $r\ne r_1$ for arbitrary $\varphi\in Aut(A_1)$.
 \end{lemma}
Here we assume that ${\rm deg} b(x)=-\infty $ if $b(x)=0$.

\begin{proof}
Let us assume that there is $\varphi\in Aut(A_1)$ such that at
$r>{\rm deg} a(x)+8$ we have $\varphi(L(r))= L(r_1)$ for some $r\ne r_1$.
Let
$$
 \varphi(x)=q_n(x)\partial_x^n+\dots+q_0(x),\qquad
 \varphi(\partial_x)=p_m(x)\partial_x^m+\dots+p_0(x),
$$
where $q_j,p_s$ are some polynomials. First consider the case $n=0$. If $n=0$,
then $m=1$ otherwise the operator $\varphi(L(r))$ has order greater than four.
Further,
$$
 \varphi(a(x)\partial_x^2+b(x)\partial_x)=
 a(q_0(x))(p_1(x)\partial_x+p_0(x))^2+b(q_0(x))(p_1(x)\partial_x+p_0(x))=
$$
$$
 a(q_0(x))p_1^2(x)\partial_x^2+
 a(q_0(x))(p_1(x)p_1'(x)+p_0(x)+b(q_0(x))p_1(x))\partial_x+
$$
$$
 +a(q_0(x))p_0(x)+b(q_0(x))+b(q_0(x))p_0(x).
$$
From our assumption it follows that
$$
 a(q_0(x))p_1^2(x)=a(x),\qquad
 a(q_0(x))(p_1(x)p_1'(x)+p_0(x))+b(q_0(x))p_1(x)=b(x).
$$
Hence from the first identity we get that $p_1(x)$ is a constant and $q_0(x)$
 is a linear function. From the second identity we get that $p_0(x)$
 is a constant, otherwise the degree of the left hand side is greater than the degree
 of the right hand side.
Thus
$$
 \varphi(x)=s_1x+s_2,\qquad \varphi(\partial_x)=s_3\partial_x+s_4,\quad
 s_j\in{\mathbb C}.
$$
From this we obtain $\varphi(L(r))\ne L(r_1).$

Let us consider the general case $n\ne 0$. We have the following
identities for orders of differential operators
$$
 {\rm ord} \varphi(a(x)\partial_x^2)=n {\rm deg} a(x)+2m,\quad
 {\rm ord} \varphi(b(x)\partial_x)=n {\rm deg} b(x)+m,\quad
 {\rm ord} \varphi(c_r(x))=rn.
$$
Let us note that
$$
 {\rm ord} \varphi(a(x)\partial_x^2)= {\rm ord} \varphi(c_r(x)),
$$
for otherwise, since
${\rm ord} \varphi(a(x)\partial_x^2)> {\rm ord} \varphi(b(x)\partial_x)$
we have
$$
 {\rm ord} \varphi(a(x)\partial_x^2+b(x)\partial_x+c_r(x))=
 {\rm ord} \varphi(a(x)\partial_x^2+c_r(x))=
 \max\{rn, n\deg a(x)+2m\} \geq r,
$$
and therefore ${\rm ord}\varphi(L(r))\geq 2r>4,$ a contradiction.
 Thus,
\begin{equation}\label{eq19}
 n{\rm deg} a(x)+2m=rn.
\end{equation}
By direct calculations one can check that
$$
 {\rm ad}(-x)^3(L(r))=[[[L(r),x],x],x]=24a^2(x)\partial_x+12a(x)b(x)+12a(x)a'(x),
$$
hence
$$
 {\rm ord}\varphi({\rm ad}(-x)^3(L(r)))=2n{\rm deg}a(x)+m.
$$
On the other hand,
$$
 \varphi({\rm ad}(-x)^3(L(r)))={\rm ad}(-\varphi(x))^3(\varphi(L(r))).
$$
We have
$$
 {\rm ord}[\varphi(L(r)),\varphi(x)]\leq n+3,\qquad
 {\rm ord}[[\varphi(L(r)),\varphi(x)],\varphi(x)]\leq 2n+2,
$$
$$
 {\rm ord}[[[\varphi(L(r)),\varphi(x)],\varphi(x)],\varphi(x)]\leq 3n+1.
$$
Thus, using (\ref{eq19}) and our assumption $r>{\rm deg}a(x)+8$,
we get
$$
 3n+1\geq {\rm ord} [{\rm ad}(-\varphi(x))^3(\varphi(L(r)))]=
 2n{\rm deg} a(x)+m=n(r+3{\rm deg}a(x))/2>\frac{n}{2}(8+4{\rm deg}a(x))
 =
$$
$$
 4n+2n{\rm deg}a(x).
$$
We get a contradiction.
\end{proof}

Hence Theorems 1.2 and 1.3 are proved.

{\bf Acknowledgements}
The authors started to discuss the problems considered in this paper
on the conference ''Around Sato's Theory on Soliton Equations'' in
Tsuda College. They are grateful to this institution for the kind
hospitality and to Atsushi Nakayashiki for the invitation.

A.E. Mironov, Sobolev Institute of Mathematics and
Novosibirsk State University, Russia; e-mail: mironov@math.nsc.ru

A.B. Zheglov, Moscow State University, Russia; e-mail:  azheglov@mech.math.msu.su

\end{document}